\documentclass[a4paper,UKenglish]{lipics}
\usepackage[T1]{fontenc}
\usepackage{amsmath}
\usepackage{amssymb}
\usepackage{amsthm}
\usepackage{mathrsfs}
\usepackage{floatflt}
\usepackage{multicol}
\usepackage{nicefrac}
\usepackage{amsfonts}
\usepackage{amscd}
\usepackage{bussproofs}
\usepackage{wasysym}
\usepackage{graphics}
\usepackage{textcomp}
\usepackage{algpseudocode}
\usepackage{color}
\usepackage{tensor}
\usepackage{algorithm}
\usepackage{url}

\algblockdefx[CASE]{CASE}{ENDCASE}[1]{\textbf{Case} {#1} \textbf{:} }{\,\,\,\,\,\,\textbf{EndCase}}
\algblockdefx[SWITCH]{SWITCH}{ENDSWITCH}[1]{\textbf{Switch} ({#1}) }{\textbf{EndSwitch}}
\algblockdefx[DEFAULT]{DEFAULT}{ENDDEFAULT}{\textbf{Default}: }{\,\,\,\,\,\,\textbf{EndDefault}}
\algblockdefx[LOOP]{LOOP}{ENDLOOP}[1]{\textbf{loop} ({#1})\{  }{\}}
\algblockdefx[FUNCTION]{FUNOPEN}{FUNCLOSE}[2]{\textbf{def }{#1}\textbf{ in }{(#2)}\{}[1]{\}\textbf{out}({#1})}

\algblockdefx[MULTISWITCH]{MULTISWITCH}{ENDMULTISWITCH}[2]{\textbf{MultiSwitch} ({#1}$\times${#2}) }{\textbf{EndSwitch}}

\newcommand{\DEFAULTLINE}[1]{\State \textbf{default:} }

\newcommand{\oursystem}{{\mathbf{iSAPP}}}

\newcommand{\some}[3]{#1_{#2}, \ldots, #1_{#3}}
\newcommand{\ssome}[4]{#1_{#2} #4 \ldots #4 #1_{#3}}
\newcommand{\many}[2]{\some{#1}{1}{#2}}
\newcommand{\smany}[3]{\ssome{#1}{1}{#2}{#3}}

\newcommand{\Alphabet}{\Sigma}
\newcommand{\letone}{a}

\newcommand{\regone}{\mathbf{r}}
\newcommand{\opone}{\mathtt{op}}

\newcommand{\Functions}{\mathtt{Functions}}
\newcommand{\definedFunctions}{\mathtt{DefinedFunctions}}
\newcommand{\deffunction}[4]{ \textbf{def }{#1}\textbf{ in }{(#2)}\,\,\{{#3}\}\,\,\textbf{out}{(#4)} }
\newcommand{\stackone}{\mathtt{S}}
\newcommand{\termstack}[1]{\langle#1\rangle}
\newcommand{\isempty}[1]{\mathtt{isempty?}(#1)}

\newcommand{\Booleans}{\mathtt{Boolean Exp}}
\newcommand{\true}{\mathsf{true}}
\newcommand{\false}{\mathsf{false}}
\newcommand{\boolexpr}{b}
\newcommand{\boolop}{\mathtt{op?}}
\newcommand{\rand}{\mathtt{rand}}

\newcommand{\Expressions}{\mathtt{Expressions}}
\newcommand{\exprzero}{e}
\newcommand{\htop}[1]{\mathtt{top}({#1})}

\newcommand{\Commands}{\mathtt{Commands}}
\newcommand{\comzero}{C}
\newcommand{\iassign}[2]{#1:=#2}
\newcommand{\iskip}{\textbf{skip}}
\newcommand{\iloop}[2]{\textbf{loop}\,{#1}\,\{{#2}\}}
\newcommand{\iifthenelse}[3]{\textbf{If }\,#1\,\textbf{ Then
  }\,#2\,\textbf{ Else }\,#3}
\newcommand{\pop}[1]{\textbf{pop}({#1})}
\newcommand{\push}[1]{\textbf{push}(#1)}
\newcommand{\call}[1]{\textbf{call}(#1)}

\newcommand{\test}[1]{{\textbf{equal?}(#1)}}

\newcommand{\matrixsubstitution}[3]{ {#1} \xleftarrow{#2}{#3} }

\newcommand{\termlist}[1]{\langle#1\rangle}

\newcommand{\plus}{+}

\newcommand{\matrixID}[0]{\mathbf{I}}
\newcommand{\vectorID}[0]{\mathbf{V}^{\mathbf{I}}}
\newcommand{\vectorA}[0]{\mathbf{V}^{\mathbf{A}}}
\newcommand{\vectorO}[0]{\mathbf{V}^{\mathbf{0}}}
\newcommand{\matrixzero}[0]{\mathbf{0}}

\newcommand{\nodep}{0}
\newcommand{\identity}{L}
\newcommand{\added}{A}
\newcommand{\multiplied}{M}

\newcommand{\valuezero}{\textit{v}}

\newcommand{\values}{\mathtt{Values}}
\newcommand{\variable}{\mathtt{Stacks}}

\newcommand{\PP}[0]{\ensuremath{\mathbf{PP}}}

\newcommand{\FP}{\ensuremath{\mathbf{FP}}}

\newcommand{\BPP}{\ensuremath{\mathbf{BPP}}}
\newcommand{\PSPACE}{\ensuremath{\mathbf{PSPACE}}}
\newcommand{\LOGSPACE}{\ensuremath{\mathbf{LOGSPACE}}}

\newcommand{\sigmazero}{\sigma_0}
\newcommand{\statefunc}{\sigma}
\newcommand{\statefuncone}{\sigma_1}
\newcommand{\statefunctwo}{\sigma_2}

\newcommand{\funone}{f}

\newcommand{\bigvar}{S}

\newcommand{\boolzero}{b}
\newcommand{\boolone}{b_1}

\newcommand{\exprone}{e_1}

\newcommand{\comone}{C_1}
\newcommand{\comtwo}{C_2}

\newcommand{\constone}{c}

\newcommand{\indexone}{i}
\newcommand{\indextwo}{j}
\newcommand{\indexthree}{k}
\newcommand{\indexfour}{p}

\newcommand{\row}{\indexone}
\newcommand{\column}{\indextwo}

\newcommand{\probone}{\alpha}
\newcommand{\probtwo}{\beta}

\newcommand{\multipolyone}{P}
\newcommand{\multipolytwo}{Q}
\newcommand{\multipolythree}{R}

\newcommand{\polyone}{p}
\newcommand{\polytwo}{q}
\newcommand{\polythree}{r}
\newcommand{\polyfour}{s}

\newcommand{\defeq}{::=}

\newcommand{\closure}[1]{{#1}^\cup}

\newcommand{\mergedown}[2]{ {#1}^{\downarrow {#2}} }
\newcommand{\vectorasjust}[2]{ \left[\stackrel{#1}{#2}\right] }

\newcommand{\matrixone}{\mathbf A}
\newcommand{\matrixtwo}{\mathbf B}

\newcommand{\vectorzero}{{V}}

\newcommand{\vectortwo}{{U}}

\newcommand{\polabstr}[1]{ \lceil {#1} \rceil}

\newcommand{\realint}[2]{\mathbb{R}_{[0,1]}}

\newcommand{\artred}{\rightarrow_a}
\newcommand{\bolred}[1]{\rightarrow_b^{{#1}}}
\newcommand{\comred}[1]{\rightarrow_c^{{#1}}}
\newcommand{\comdistr}[0]{\rightarrow_\distrone }

\newcommand{\polyunion}[2]{ ({#1}\oplus{#2})  }

\newcommand{\polynormform}[1]{ \mathtt{{#1}} }

\newcommand{\subst}[3]{#1[#2/#3]}
\newcommand{\distrone}{\mathscr{D}}
\newcommand{\distrtwo}{\mathscr{E}}

\newcommand{\multi}[1]{\overline{#1}}

\newcommand{\derone}{\pi}

\newcommand{\size}[1]{|#1|}

\newcommand{\taper}{\mathbf{tape}_r}
\newcommand{\tapel}{\mathbf{tape}_l}
\newcommand{\tapeh}{\mathbf{tape}_h}
\newcommand{\statemachine}{\mathbf{M_{state}}}

\newcommand{\dimensionone}{\numeone}
\newcommand{\dimensiontwo}{\numetwo}

\newcommand{\numeone}{n}
\newcommand{\numetwo}{m}

\newcommand{\NN}{\mathbb{N}}

\newcommand{\pair}[2]{\langle #1,#2 \rangle  }

\newcounter{number}

\title{{A polytime complexity analyser for Probabilistic Polynomial Time over imperative stack programs.}}
\titlerunning{{$\oursystem$, a complete polytime complexity verifier tool for PP.}}
\author[1]{J.Y. Moyen}
\author[2]{P. Parisen Toldin}
\affil[1]{{LIPN, UMR 7030, CNRS, Universit\`{e} Paris 13
 	  F-93430 Villetaneuse, France.
         \textit{Jean-Yves.Moyen@lipn.univ-paris13.fr}}}
\affil[2]{Dipartimento di Scienze dell'Informazione, Universit\`a di Bologna
            \'Equipe FOCUS, INRIA Sophia Antipolis
            Mura Anteo Zamboni 7, 40127 Bologna, Italy.
            \textit{parisent@cs.unibo.it}}
            
\Copyright{J.Y. Moyen and P. Parisen Toldin}

\subjclass{Theory, Verification}

\begin{document}

\maketitle

\begin{abstract}
  We present $\oursystem$ (Imperative Static Analyser for Probabilistic
  Polynomial Time), a complexity verifier tool that is sound and
  extensionally complete for the Probabilistic Polynomial Time ($\PP$)
  complexity class. $\oursystem$ works on an imperative
  programming language for stack machines.
  The certificate of polynomiality can be built in
  polytime, with respect to the number of stacks used.
\end{abstract}
 
\section{Introduction}

One of the crucial problem in program analysis is to understand how much
time it takes a program to complete its run. Having a bound on running
time or on space consumption is really useful, specially in fields of
information technology working with limited computing power.
Solving this problem for every program is well known to be
undecidable. The best we can do is to create an analyser for a
particular complexity class able to say ``yes'', ``no'', or ``don't
know''. Creating such an analyser can be quite easy: the one saying
every time ``don't know'' is a static complexity analyser. The most
important thing is to create one that answers ``don't know'' the minimum
number of time as possible.

We try to combine this problem with techniques derived from Implicit
Computational Complexity (ICC).  Such research field combines
computational complexity with mathematical logic, in order to give
machine independent characterisations of complexity classes. ICC has
been successfully applied to various complexity classes such as
$\FP$~\cite{Bellantoni1992, Leivant1993, BMMTCS},
$\PSPACE$~\cite{Leivant1995}, $\LOGSPACE$~\cite{Jones1999}.

ICC systems usually work by restricting the constructions allowed in a
program. This \emph{de facto} creates a small programming language whose
programs all share a given complexity property (such as computing in
polynomial time). ICC systems are normally \textbf{extensionally
  complete}: for each function computable within the given complexity
bound, there exists one program in the system computing this
function. They also aim at \textbf{intentional completeness}: each
program computing within the bound should be recognised by the
system. Full intentional completeness, however, is undecidable and ICC
systems try to capture as many programs as possible (that is, answer
``don't know'' as little time as possible).

Having an ICC system characterising a complexity class $\mathcal{C}$ is
a good starting point for developing a static complexity
analyser.
There is a large literature on static analysers for complexity
bounds. We develop an analysis recalling methods from \cite{Jones2009a,
  BenAmram2008, Kristiansen2005}. Comparatively to these approaches our
system works with a more concrete language of stacks, where variables,
constants and commands are defined; we are also sound and complete with
respect to the Probabilistic Polynomial time complexity class ($\PP$) \cite{Gill77}.

We introduce a probabilistic variation of the \textsc{Loop}
language. Randomised computations are nowadays widely used and most of
efficient algorithms are written using stochastic information. There are
several probabilistic complexity classes and $\BPP$ (which stands for
Bounded-error Probabilistic Polytime) \cite{Gill77} is considered close to the informal
notion of feasibility. Our work would be a first step into the direction
of being able to capture real feasible programs solving problems in
$\BPP$ ($\BPP \subseteq \PP$) \cite{Gill77}.

Similar work has been done in~\cite{PPTUDL2011} with characterisation of
complexity class $\PP$; This work gives a characterisation of complexity
class $\PP$ by using a functional language with safe recursion as in
Bellantoni and Cook~\cite{Bellantoni1992}.

Our system is called $\oursystem$, which stands for \emph{Imperative
  Static Analyser for Probabilistic Polynomial Time}. It works on a
prototype of imperative programming language based on the \textsc{Loop}
language~ \cite{Meyer1967}. 
The main purpose of this paper is to present a minimal
probabilistic polytime certifier for imperative programming
languages. 

Following ideas from~\cite{Jones2009a, BenAmram2008} we ``type''
commands with matrices, while we do not type expressions since they have constant size.
The underlying idea
is that these matrices express a series of polynomials bounding
the size of stacks, with respect to their input size. The algebra on
which these matrices and vectors are based is a finite (more or less
tropical) semi-ring.
%
%
\section{Stacks machines}
We study \emph{stacks machines}, a generalisation of the classical
counters machines. Informally, a stacks machine work with \emph{letters}
belonging to a finite alphabet and \emph{stacks} of letters. Letters can
be manipulated with \emph{operators}. Typical alphabet include the
binary alphabet $\{0, 1\}$ or the set \texttt{long int} of 64 bits
integers. On the later, typical operators are $+$ or $*$.

Each machine has a finite number of registers that may hold letters and
a finite number of stacks that may hold stacks. Tests can be made
either on registers and letters (with boolean operators) or to check
whether a given stack holds the empty stack. There are only bounded
(\texttt{for}) loops which are controlled by the size of a given
stack. That is, it is more alike a \texttt{foreach} (element in the
stack) loop. Since there are only bounded loops (and no \texttt{while}),
this \emph{de facto} limits the language to primitive recursive
functions. In this way, stack machines are a generalisation of the
classical \textsc{Loop} language~\cite{Meyer1967}. 
Since our analysis is compositional, 
we add also functions to the language; their certificates can be computed
separately and plugged in the right place when a call is performed.
\subsection{Syntax and Semantics}
We denote $\termstack{}$ the empty stack and
$\termstack{\smany{\letone}{n}{}}$ the stack with $n$ elements and
$\letone_1$ at top.

\begin{definition}\label{def:gramdefinition}
  A \emph{stacks machine} consists in:
  \begin{itemize}
  \item a finite alphabet $\Alphabet = \{\many{\letone}{n}\}$ containing
    at least two values $\true$ and $\false$;
  \item a finite set of operators, $\opone_i$, of type $\Alphabet^n \to
    \Alphabet$, containing at least a 0-ary operator $\rand$, operators
    whose co-domain is $\{\true, \false\}$ are \emph{predicates} noted
    $\boolop$;
  \item a finite set of registers $\regone$ and
    stacks, $\stackone_j$ (the empty stack is noted $\termstack{}$);
  \item and a program written in the following syntax:
\begin{align*}
  \boolexpr \in \Booleans \defeq & \true \,|\, \false \,|\,
  \boolop(\many{\exprzero}{n}) \,|\, \rand() \,|\, \isempty{\stackone}\\
  \exprzero \in \Expressions \defeq & \constone \,|\, \regone \,|\,
  \opone(\many{\exprzero}{n}) \,|\, \htop\stackone \\
  %
  %
  \comzero \in \Commands \defeq & \iskip \,|\,
  \iassign{\regone}{\exprzero} \,|\,
  \iassign{\stackone_1}{\stackone_2}
  \,|\, \iassign{\stackone}{\termstack{\smany{\constone}{n}{}}} \,|\,
   \iassign{\stackone_k}{\call{\funone,\smany{\stackone}{n}{}}}\\
  &  \,|\, \pop{\stackone}\,|\, \push{\exprzero,\stackone} \,|\, \comzero;\comzero \,|\,
  \iifthenelse{\boolexpr}{\comzero}{\comzero} \,|\,
  \iloop{\stackone}{\comzero}\\
  \funone \in \Functions \defeq & \deffunction{f}{\smany{\stackone}{n}{}}{\comzero}{\stackone_j}
\end{align*}
\end{itemize}
\end{definition}

 Note that registers may not appear directly in booleans expressions to
 avoid dealing with the way non-booleans values are interpreted as
 booleans. 
 However, it is easy to define a unary predicate which,
 \emph{e.g.} sends $\true$ to $\true$ and every other letter to $\false$
 to explicitly handle this.

 Expressions always return letters (content of registers) while commands
 modify the state but do not return any value. $\htop\ $ does not destruct
 the stack but simply returns its top element while $\pop\ $ remove the top
 element from the stack but does not return anything. It is also possible
 to assign constant stack to a stack.

The $\isempty$ predicate returns $\true$ if and only if the stack
given in argument holds the empty stack and $\false$ otherwise.  The
$\iloop{\stackone}{\comzero}$ commands executes $\comzero$ as many time
as the size of $\stackone$. 
Moreover, $\stackone$ may not appear in
 $\comzero$. It is, however, possible to make a copy beforehand if the
 content is needed within the loop. Finally, we give the possibility to have function call.
 The command $\call{\funone,\smany{\stackone}{n}{}}$ call the function $\funone$ passing the actual arguments $\smany{\stackone}{n}{}$
 and finally return the result stored in the stack $\stackone_j$.
 

\subsection{Complexity}
The set of operators is not specified and may vary from one stacks
machine to another (together with the alphabet). This allows for a wide
variety of settings parametrised by these.  Typical alphabets are the
binary one ($\{\true, \false\}$), together with classical boolean
operators (\texttt{not}, \texttt{and},~\ldots); or the set \texttt{long
  int} of 64 bits integers with a large number of operators such as
\texttt{+}, \texttt{*}, \texttt{<},~\ldots Since there is only a finite
number of letters and operators all have the alphabet as domain and
co-domain, there is only a finite number of operators at each arity. So,
without going deep into details, it makes sense to consider that each
operator take a constant time to be computed.
 More precisely, each operator can be computed within a time
 bounded by a constant. Typically, on \texttt{long int}, \texttt{+} can
 be computed in 64 elementary (binary) additions and \texttt{*} takes a
 bit more operations but is still done in bounded time.
 
 Thus, in order to simplify the study, we consider that operators are
 computed in constant time and we do not need to take individual
 operators into account when bounding complexity. It is sufficient to
 consider the number of operators.
 
 The only thing that is unbounded is the size (length) of stacks. Thus,
 if one want, \emph{e.g.} to handle large integers (larger than the size
 of the alphabet), one has to encode them within stacks. The most obvious
 ways being the unary representation (a number is represented by the size
 of a stack) and the binary one (a stack of \texttt{0} and \texttt{1} is
 interpreted as a binary number with least significant bit on
 top). Obviously, any other base can be use. In each case, addition (and
 multiplication) has to be defined for this representation of ``large
 integers'' with the tools given by the language (loops). Of course,
 encoding unbounded value is crucial in order to simulate arbitrary
 Turing Machines (or even simply \textsc{Ptime} ones) and is thus
 required for the completeness part of the result.
 
 Note that copying a whole stack as a single instruction is a bit
 unrealistic as it would rather takes time proportional to the size of
 the stack. However, since each stack will individually be bounded in
 size by a polynomial, this does not hampers the polynomiality of the
 program. A clever implementation of stacks with pointers (\emph{i.e.} as
 lists) will also allow copy of a whole stack to be implemented as copy
 of a single pointer, an easy operation.

 Since the language only provides bounded loops whose number of execution
can be (dynamically) known before executing them, only primitive
recursive functions may be computed. This may look like a big
restriction but actually is quite common within classical ICC results on
\textsc{Ptime}. Notably, Cobham~\cite{Cob65} or Bellantoni and
Cook~\cite{Bellantoni1992} both work on restrictions of the primitive
recursion scheme; Bonfante, Marion and Moyen~\cite{BMMTCS} split the
size analysis (quasi-interpretation) from a termination analysis
(termination ordering) which also characterise only primitive recursive
programs; and lastly Jones and Kristiansen~\cite{Jones2009a}, on which
this work is directly based, use the \textsc{Loop} language which also
allows only primitive recursion.
 
 Since loops are bounded by the size of stacks, it is sufficient to bound
 the size of stacks in order to bound the time complexity of the
 program. Indeed, if each stack has a size smaller than $\polyone$ and the
 program has never more than $k$ nested loops, then its runtime cannot be
 larger than $\polyone^k$. Similarly, in the original $mwp$ calculus of Jones
 and Kristiansen, it was sufficient to bound the value of stacks in
 order to bound the runtime of programs (for the same reasons). Note that
 to have a large number of iterations, one first has to create a stack of
 large size, that is when bounding the number of iterations stacks are
 considered \emph{de facto} as unary numbers.
 
For each stack, we keep the dependencies it
has from the other stacks. For example, after a copy
($\iassign{\stackone_1}{\stackone_2}$), the size of $\stackone_1$ is the
same as the size of $\stackone_2$. Keeping precise dependencies is not
manageable, so we only keep the \emph{shape} of the dependence
(\emph{e.g.} the degree with which it appear in a polynomial). These
shapes are collected in a vector (for each stack) and combining all of
them gives a matrix certificate expressing the size of the output
stacks relatively to the size of the input stacks.
The matrix calculus we obtain for the certificates is
compositional. This allows for a modular approach of building
certificates.

\section{Algebra}
 
Before going deeply in explaining our system, we need to present the
algebra on which it is based.  $\oursystem$ is based on a finite algebra
of values. The set of scalars is $\values=\{\nodep , \identity , \added
, \multiplied\}$ and these are ordered in the following way $\nodep <
\identity < \added < \multiplied$. The idea behind these elements is to
express how the value of stacks influences the result of an
expression. $\nodep$ expresses no-dependency between stack and
result; $\identity$ (stands for ``Linear'') expresses that the result
linearly depends with coefficient $1$ from this stack.  $\added$
(stands for ``Additive'') expresses the idea of generic affine
dependency.  $\multiplied$ (stands for ``Multiplicative'') expresses the
idea of generic polynomial dependency.
 
 
We define sum, multiplication and union in our algebra as expressed in
Table~\ref{tb:multadddef}. The reader will immediately notice that
$\identity + \identity$ gives $\added$, while $\identity \cup \identity$
gives $\identity$ The operator $\cup$ works as a maximum.
\begin{table}
  \begin{center}
    \begin{tabular}{ c | c | c | c | c  }
      $\times$ &  \nodep & \identity & \added & \multiplied \\ \hline
      \nodep & \nodep & \nodep & \nodep & \nodep \\ \hline
      \identity& \nodep & \identity & \added & \multiplied \\ \hline
      \added & \nodep & \added &\added &\multiplied  \\ \hline\
      \multiplied & \nodep &\multiplied &\multiplied &\multiplied \\
    \end{tabular}
    \hspace{5pt}
    \begin{tabular}{ c | c | c | c | c  }
      \plus &  \nodep & \identity & \added & \multiplied \\ \hline
      \nodep & \nodep & \identity & \added & \multiplied \\ \hline
      \identity&\identity &\added &\added &\multiplied \\ \hline
      \added & \added & \added &\added &\multiplied  \\ \hline\
      \multiplied & \multiplied &\multiplied &\multiplied &\multiplied \\
    \end{tabular}
    \hspace{5pt}
    \begin{tabular}{ c | c | c | c | c  }
      $\cup$ &  \nodep & \identity & \added & \multiplied \\ \hline
      \nodep & \nodep & \identity & \added & \multiplied \\ \hline
      \identity&\identity &\identity &\added &\multiplied \\ \hline
      \added & \added & \added &\added &\multiplied  \\ \hline\
      \multiplied & \multiplied &\multiplied &\multiplied &\multiplied \\
    \end{tabular}
    \caption{Multiplication, addition and union of values}
    \label{tb:multadddef}
  \end{center}
\end{table}
Over this semi-ring we create a module of matrices, where values are
elements of $\values$. We define a partial order $\le$ between matrices
of the same size as component wise ordering. Particular matrices are
$\matrixzero$, the one filled with all $\nodep$, and $\matrixID$, the
identity matrix, where elements of the main diagonal are $\identity$ and
all the others are $\nodep$. If $\valuezero\in\values$, a particular vector is
$\mathbf{\vectorzero}^\valuezero_\indexone$ that is a column vector full of zeros and having
$\valuezero$ at $\indexone$-th row.
Multiplication and addition between matrices work as usual\footnote{That
  is: $(\matrixone + \matrixtwo)_{\indexone, \indextwo} =
  \matrixone_{\indexone, \indextwo} + \matrixtwo_{\indexone, \indextwo}$
  and $(\matrixone \times \matrixtwo)_{\indexone, \indextwo} = \sum
  \matrixone_{\indexone, \indexthree} \times \matrixtwo_{\indexthree,
    \indextwo}$ } and we define point-wise union between matrices:
$(\matrixone \cup \matrixtwo)_{\indexone, \indextwo} =
\matrixone_{\indexone, \indextwo} \cup \matrixtwo_{\indexone,
  \indextwo}$. Notice that $\matrixone \cup \matrixtwo \leq \matrixone +
\matrixtwo$. As usual, multiplication between a value and a matrix
corresponds to multiplying every element of the matrix by that value.
 
%
We can now move on and present some new operators and properties of
matrices. Given a column vector $\vectorzero$ of dimension
$\dimensionone$, a matrix $\matrixone$ of dimension $\dimensionone
\times \dimensiontwo$ an index $\indexone$ ($\indexone \le
\dimensiontwo$), we indicate with
$\matrixsubstitution{\matrixone}{\indexone}{\vectorzero}$ a substitution
of the $\indexone$-th column of the matrix $\matrixone$ with the vector
$\vectorzero$.
 
Next, we need a closure operator. The ``union closure'' is the union of
all powers of the matrix: $\closure{\matrixone} = \bigcup_{\indexone\ge
  0} \matrixone^\indexone$. It is always defined because the set of
possible matrices is finite.
We will need also a ``merge down'' operator. Its use is to propagate
the influence of some stacks to some other and it is used to correctly
detect the influence of stacks controlling loops onto stacks modified
within the loop (hence, we can also call it ``loop correction''). The
last row and column of the matrix is treated differently because it will
be use to handle constants and not stacks. In the following, $\dimensionone$ is size of the vector, $k<n$ and $j<n$.
%
%
\begin{itemize}
\item 
    $(\mergedown{\vectorzero}{\indexthree,\dimensionone})_\indexone = \vectorzero_\indexone$ 
\item 
   $(\mergedown{\vectorzero}{\indexthree, \column})_{\indexthree} =
   \begin{cases}
     \multiplied & \text{ if $\exists \indexfour < \dimensionone,
       \indexfour \neq \indexthree$ such that $\vectorzero_{\indexfour}  
       \neq \nodep$} \\
     \nodep & \text{ otherwise and $\vectorzero_{\dimensionone} =
       \nodep$}  \\
     \identity & \text{ otherwise and $\vectorzero_{\dimensionone} =
       \identity$} \\
     \added & \text{ otherwise and $\vectorzero_{\dimensionone} \ge
       \added$} \\
     \vectorzero_{\indexthree} & \text{ otherwise}\\
   \end{cases}$
 \item $ 
    (\mergedown{\vectorzero}{\indexthree, \column})_{\row} =
    \begin{cases}
      \nodep &\text{ if $i=n$ } \\
      \multiplied & \text{ if $i\neq j$, $\vectorzero_{\row} \neq \nodep$ and
        $\vectorzero_{\column} \neq \nodep$}\\
      \vectorzero_{\row} & \text{ otherwise}
    \end{cases}$ 
 \end{itemize}
 
 In the following we will use a slightly different notation. Given a matrix $\matrixone$ and an index $\indexthree$, 
 $\mergedown{\matrixone}{\indexthree}$ is the matrix obtained by applying the previous definition of merge down on 
 each column of $\matrixone$.
  Formally, if $\vectorzero$ is the $\indextwo$-th column of $\matrixone$, then $\indextwo$-th column of
 $\mergedown{\matrixone}{\indexthree}$ is $\mergedown{\vectorzero}{\indexthree,\indextwo}$.
 
 Finally, the last operator that we are going to introduce is the ``re-ordering'' operator. Given a vector $\vectorzero$ we write
 $\vectorasjust{\vectorzero}{1\rightarrow \indexone,\ldots, n\rightarrow \indextwo}$
 to indicate that the result is a vector whose rows are permuted. The first raw goes in the $\indexone$-th position and so on till the $n$-th to the $\indextwo$-th. In order to use a short notation, if a row is flowing in its same position, then we don't explicit it.
 Formally, if $\vectortwo =  \vectorasjust{\vectorzero}{1\rightarrow \indexone,\ldots, n\rightarrow \indextwo}$, then:
 $
 \vectortwo_\indexfour = \sum_\indexthree \vectorzero_\indexthree \,|\, \indexthree\rightarrow\indexfour
 $.
 
 So, in case two or more rows clash on the same final row, we perform a sum between the values. This operator is used for certificate the function calls. Indeed we have to connect the formal parameters with the actual parameters. Therefore, we have to permute the result of the function in order to keep track where the actual parameters has been substituted in place of the formal parameters.
 
\section{Multipolynomials and abstraction}\label{sc:multipolynomial}
 
  We can now proceed and introduce another fundamental concept for
  $\oursystem$: multipolynomials. This concept was firstly presented in
  \cite{MMDiceMultipol}. A multipolynomial represents real bounds and its
  abstraction is a matrix. In the following, we assume that every
  polynomial have positive coefficients and it is in the canonical form.
 
  First, we need to introduce some operator working on polynomials. 
%
%
%
  
  \begin{definition}[Union of polynomial]\label{def:unionmultipolinomial}
    Be $\polynormform{\polyone}$, $\polynormform{\polytwo}$ the
    canonical form of the polynomials $\polyone$, $\polytwo$ and let $\polythree,\polyfour$ polynomials, $\alpha,\beta$ natural numbers, 
    we define the operator $\polyunion{\polyone}{\polytwo}$ over
    polynomials in the following way:
    \[
     \polyunion{\polyone}{\polytwo} =        \begin{cases}
     					      \max{(\alpha,\beta)} + \polyunion{\polythree}{\polyfour} & \text{if $\polynormform\polyone = \alpha +\polythree$ and $\polynormform\polytwo = \beta + \polyfour$.}\\
					      \max{(\alpha,\beta)}\cdot X_i + \polyunion{\polythree}{\polyfour} & \text{otherwise if $\polynormform\polyone = \alpha X_i+\polythree$ and $\polynormform\polytwo = \beta X_i + \polyfour$.}\\
					      \polynormform\polyone + \polynormform\polytwo & \text{otherwise}
                                             \end{cases}\
    \]
  \end{definition}
  
  Let's see some example. Suppose we have
  these two polynomials: ${X_1} + 2{X_2} +
  3{X_4}^2{X_5}$ and ${X_1}+ 3{X_2} +
  3{X_4}^2{X_5}+ {X_6}$. Call them, respectively
  $\polyone$ and $\polytwo$. We have that
  $\polyunion{\polyone}{\polytwo}$ is ${X_1} + 3{X_2} +
  3{X_4}^2{X_5} + {X_6}$.
 
 First we need to introduce the concept of abstraction of
 polynomial. Abstraction gives a vector representing the shape of our
 polynomial and how variables appear inside it.
 
 \begin{definition}[Abstraction of polynomial]\label{def:abstractionmultipolynomials}
   Let $\polyone(\multi{X})$ a polynomial over $\numeone$ variables,
   $\polabstr{{\polyone(\multi{X})} }$ is a column vector of size
   $\numeone+1$ such that:
  \begin{itemize}
  \item If $\polyone(\multi{X})$ is a constant $\constone > 1$, then
    $\polabstr{\polyone(\multi{X})}$ is
    $\vectorA_n$
  \item Otherwise if $\polyone(\multi{X})$ is a constant $0$ or $1$, then
    $\polabstr{\polyone(\multi{X})}$ is respectively
    $\vectorO_n$ or $\vectorID_n$.
  \item Otherwise if $\polyone(\multi{X})$ is ${X}_\indexone$, then
    $\polabstr{\polyone(\multi{X})}$ is
    $\vectorID_\indexone$.
  \item Otherwise if $\polyone(\multi{X})$ is $\alpha{X}_\indexone$ (for
    some constant $\alpha>1$), then $\polabstr{\polyone(\multi{X})}$
    is
    $\vectorA_\indexone$.
  \item Otherwise if $\polyone(\multi{X})$ is
    $\polytwo(\multi{X})+\polythree(\multi{X})$, then
    $\polabstr{\polyone(\multi{X})}$ is
    $\polabstr{\polytwo(\multi{X})}+\polabstr{\polythree(\multi{X})}$.
  \item Otherwise, $\polyone(\multi{X})$ is
    $\polytwo(\multi{X})\cdot\polythree(\multi{X})$, then
    $\polabstr{\polyone(\multi{X})}$ is
    $\multiplied\cdot\polabstr{\polytwo(\multi{{X}})} \cup
    \multiplied\cdot\polabstr{\polythree(\multi{{X}})}$.
  \end{itemize}
 \end{definition}
 
 Size of vectors is $\numeone+1$ because $\numeone$ cells are needed for
 keeping track of $\numeone$ different variables and the last cell is the
 one associated to constants.  We can now introduce multipolynomials and
 their abstraction.
 
 \begin{definition}[Multipolynomials]
   A multipolynomial is a tuple of polynomials. Formally $\multipolyone =
   (\polyone_1,\ldots,\polyone_\numeone)$, where each
   $\polyone_\indexone$ is a polynomial.
 \end{definition}
 
 In the following, in order to refere to a particular polynomial of a multipolynomial we will use an index. So, $\multipolyone_\indexone$ refers to the $\indexone$-th polynomial of $\multipolyone$.
 Now that we have introduced the definition of multipolynomials, we can go on and present two foundamental operation on them: sum and composition.
 
\begin{definition}[Sum of multipolynomials]
 Given two multipolynomials $\multipolyone$ and $\multipolytwo$ over the same set of variables, we define addition in the following way:
 $ (\multipolyone \oplus \multipolytwo)_\indexone =  \polyunion{\multipolyone_\indexone}{\multipolytwo_\indexone} $.
\end{definition}

\begin{definition}[Composition of multipolynomial]
 Given two multipolynomials $\multipolyone$ and $\multipolytwo$ over the same set of variables, the composition of two multipolynomials is defined as the composition component-wise of each polynomial. Formally we define composition in the following way:
 $(\multipolyone\cdot\multipolytwo) = \multipolytwo_1\cdot\multipolyone_1,\ldots,\multipolytwo_\numeone\cdot\multipolyone_\numeone$.
\end{definition}

 Abstracting a multipolynomial naturally gives a matrix where each column
 is the abstraction of one of the polynomials.
 
 \begin{definition}
   Let $\multipolyone$ be a
   multipolynomial, its abstraction $\polabstr{\multipolyone}$ is a
   matrix where the $\indexone$-th column is the vector
   $\polabstr{\multipolyone_\indexone}$.
 \end{definition}
 
 In the following, we use polynomials to bound size of single
 stacks. Since handling polynomials is too hard (\emph{i.e.}
 undecidable), we only keep their abstraction. Similarly, we use
 multipolynomials to bound the size of all the stacks of a program at
 once. Again, rather than handling the multipolynomials, we only work
 with their abstractions.

\section{Typing and certification}

We presented all the ingredients of $\oursystem$ and we are ready to
introduce certifying rules.  Certifying rules, in figure
\ref{def:typingRulesCommands}, associate at every command a matrix. We suppose to have $\numeone-1$ stacks.
Notice how expressions are not typed; indeed, we 
don't need to type them because their size is fixed.

\begin{figure*}[htbp]
\begin{center}
\fbox{
\begin{minipage}{.97\textwidth}\begin{scriptsize}
$$
 \AxiomC{$n>1$ }\RightLabel{\textsc{(Const-A)}}
 \UnaryInfC{$\vdash \iassign{\stackone}{\termstack{\smany{\constone}{n}{}}} : 
 \matrixsubstitution{\matrixID}{\indexone}{\vectorA_\numeone}$} \DisplayProof
 \hspace{10pt}
 \AxiomC{}\RightLabel{\textsc{(Const-L)}}
 \UnaryInfC{$\vdash \iassign{\stackone}{\termstack{{\constone_1}}} : 
 \matrixsubstitution{\matrixID}{\indexone}{\vectorID_\numeone}$} \DisplayProof
 \hspace{10pt}
 \AxiomC{}\RightLabel{\textsc{(Const-0)}}
 \UnaryInfC{$\vdash \iassign{\stackone}{\termstack{{}}} : 
 \matrixsubstitution{\matrixID}{\indexone}{\vectorO_\numeone}$} \DisplayProof
$$
$$
\AxiomC{}\RightLabel{\textsc{(Axiom-Reg)}} \UnaryInfC{$\vdash
   \iassign{\regone}{\exprzero}:\matrixID$} \DisplayProof
   \hspace{10pt}
\AxiomC{}\RightLabel{\textsc{(Push)}}
\UnaryInfC{$\vdash \push{\exprzero, \stackone_\indexone} : \matrixsubstitution{\matrixID}{\indexone}{(\vectorID_\numeone+\vectorID_\indexone)}$}
\DisplayProof
\hspace{10pt}
\AxiomC{$  \vdash \comone : \matrixone$}
\AxiomC{$  \vdash \comtwo : \matrixtwo$}\RightLabel{\textsc{(Concat)}}
\BinaryInfC{$  \vdash \comone;\comtwo : \matrixone \times \matrixtwo$}
\DisplayProof
$$
$$
\AxiomC{$
  \vdash \comone : \matrixone $} \AxiomC{$ \forall \indexone,
  {(\closure\matrixone)}_{\indexone,\indexone} < \added $}
\RightLabel{\textsc{(Loop)}} \BinaryInfC{$ \vdash
  \iloop{\bigvar_\indexthree}{\comone} :
  \mergedown{(\closure{\matrixone})}{\indexthree}$} \DisplayProof
\hspace{10pt}
\AxiomC{$\vdash \comzero: \matrixone$}
\AxiomC{$\matrixone\le \matrixtwo$}\RightLabel{\textsc{(Subtyp)}}
\BinaryInfC{$\vdash \comzero : \matrixtwo$}
\DisplayProof
 \hspace{10pt}
 \AxiomC{} \RightLabel{\textsc{(Asgn)}} 
 \UnaryInfC{$ \vdash \stackone_\indexone:=\stackone_\indextwo: \matrixsubstitution{\matrixID}{\indexone}{\vectorID_\indextwo}
   $} \DisplayProof
  $$
  $$  
\AxiomC{$\boolone\in\Booleans$}\AxiomC{$ \vdash
  \comone:\matrixone$}\AxiomC{$
  \vdash\comtwo:\matrixtwo$}\RightLabel{\textsc{(IfThen)}}
\TrinaryInfC{$ \vdash\iifthenelse{\boolone}{\comone}{\comtwo} :
  \matrixone \cup \matrixtwo$} \DisplayProof 
\hspace{10pt}
\AxiomC{$\vdash \comzero:\matrixone$}\RightLabel{\textsc{(Fun)}}
\UnaryInfC{$\deffunction{f}{\smany{\stackone}{n}{}}{\comzero}{\stackone_j}:\matrixone$}
\DisplayProof
$$
$$
 \AxiomC{}\RightLabel{\textsc{(Skip)}} \UnaryInfC{$\vdash
   \iskip:\matrixID$} \DisplayProof
 \hspace{10pt}
 \AxiomC{}\RightLabel{\textsc{(Pop)}} 
\UnaryInfC{$ \vdash \pop\stackone : \matrixID$} \DisplayProof
\hspace{10pt}
\AxiomC{$\deffunction{f}{\smany{\stackone}{n}{}}{\comzero}{\stackone_j}:\matrixone$}\RightLabel{\textsc{(FunCall)}}
\UnaryInfC{${\stackone_i}:=\call{f,\stackone_\indexthree,\ldots,\stackone_\indexfour}: \matrixsubstitution{\matrixID}{\indexthree}{ 
\vectorasjust{\matrixone_j}{(1\rightarrow\indexthree,\ldots,n\rightarrow\indexfour )}   } $}
\DisplayProof
$$
\end{scriptsize}
\end{minipage}
}
\end{center}
\caption{Typing rules for commands and functions}
\label{def:typingRulesCommands}
\end{figure*}

These matrices tell us about the behaviour of a command and functions. 
We can think about them as certificates. Certificates for commands tell us about the correlation between
input and output stacks. Each column gives the bound of one output stack
while each row corresponds to one input stack. Last row and column
handle constants.

As example, command $\textsc{(Skip)}$ tells us that no stack is
changed. Concatenation of commands $\textsc{(Concat)}$
tells us how to find a certificate for a series of commands. The
intrinsic meaning of matrix multiplication is to ``connect'' output of
the first certificate with input of the second. In this way we rewrite
outputs of the second certificate respect to inputs of the first one.
Notice how the rule for \textsc{(Push)} does not have any hypothesis.
Indeed, this command just increase by $+1$ (a constant) the size of the stack $\stackone_\indexone$.
When there is a test, taking the union (\emph{i.e.} maximum) of the
certificates means taking the worst possible case between the two
branches.
The most interesting type rule is the one concerning the
$\textsc{(Loop)}$ command. The right premise acts as a guard: an
$\added$ on the diagonal means that there is a stack
$\stackone_\indexone$ such that iterating the loop a certain number of
time results in (the size of) $\stackone_\indexone$ depending affinely
of itself, \emph{e.g.} $\size{\stackone_\indexone} = 2 \times
\size{\stackone_\indexone}$. Obviously, iterating this loop may create an
exponential growth, so we stop the analysis immediately. Next, the union
closure used as a certificate corresponds to a worst case scenario. We
can't know if the loop will be executed 0, 1, 2, \ldots times each
corresponding to certificates $\matrixone^0, \matrixone^1, \matrixone^2,
\ldots$ Thus we assume the worst and take the union of these, that is
the union closure.  Finally, the loop correction (merge down) is here to
take into account the fact that the result will also depends on the size
of the stack controlling the loop (\emph{i.e.} the index $k$ 
is the number of the variable $\bigvar_k$ controlling the
loop).

Before start to prove the main theorems, let present some examples using the commands $\call{}$, $\iloop{}{}$.
In the following we will use integer number like $0,1,2,\ldots$ intending a constant list of size $0,1,2,\ldots$.
This should help the reader.

\begin{example}[Addition]\label{def:addition}
We are going to present the function $+$ (a shortcut for the following function). 
We can check that the analysis of this function is exactly the one expected. 
The size of the result is the sum of the sizes of the two stacks.

\begin{algorithm}
\begin{multicols}{2}
\begin{algorithmic}
\FUNOPEN{addition}{$S_1,S_2$}
  \State $\iassign{S_3}{S_2}$
  \LOOP{$S_2$}
    \State \push{$ \htop{S_3},S_1$}
    \State \pop{$S_3$}
  \ENDLOOP
\FUNCLOSE{$S_1$}
\end{algorithmic}

The associate matrix of this function is exactly what we are expecting. Indeed, the matrix is the following one: 
$\begin{bmatrix}
\identity & \nodep    &  \nodep    &  \nodep \\ 
\identity & \identity &  \identity &  \nodep \\ 
\nodep    & \nodep    &  \nodep    &  \nodep \\ 
\nodep    & \nodep    &  \nodep    &  \identity \\ 
\end{bmatrix}$
\end{multicols}
\end{algorithm}
\end{example}

\begin{example}[Multiplication]\label{def:multiplicationconstant}
 In the following we present a way to type multiplication between a number and a variable. In the following $S_2$ is multiplied by $n$
 and the result is stored in $S_1$.
 
\begin{algorithm}
\begin{multicols}{2}
 \begin{algorithmic}
  \State $\iassign{S_1}{0}$
  \LOOP{$S_2$}
    \State $S_1 = S_1 + n$
  \ENDLOOP
 \end{algorithmic}
  \vfill\columnbreak\vfill
  typed with
  $\begin{bmatrix}
    \nodep    & \nodep    & \nodep \\
    \added    & \identity & \nodep \\
    \nodep    & \nodep    & \identity \\
   \end{bmatrix}$
  
\end{multicols}
\end{algorithm}
\end{example}

\begin{example}[Multiplication]\label{def:multiplication}
In this example we show how to type a multiplication between two variables.
\begin{algorithm}
 \begin{multicols}{2}
  \begin{algorithmic}
  \FUNOPEN{multiplication}{$S_1,S_2$}
    \State $\iassign{S_3}{0}$
    \LOOP{$S_2$}
      \State $\iassign{S_3}{S_1+S_3}$
    \ENDLOOP
  \FUNCLOSE{$S_1$}  \end{algorithmic}
  \vfil\columnbreak\vfill
The loop is typed with the matrix 
$\begin{bmatrix}
  \identity & \nodep & \multiplied & \nodep \\
  \nodep    & \identity & \multiplied & \nodep \\
  \nodep   & \nodep  & \identity & \nodep  \\
  \nodep & \nodep   & \nodep  & \identity  \\
 \end{bmatrix}$.
 
So, the entire function is typed with 
$\begin{bmatrix}
  \identity & \nodep & \multiplied & \nodep \\
  \nodep    & \identity & \multiplied & \nodep \\
  \nodep   & \nodep  & \nodep & \nodep  \\
  \nodep & \nodep   & \nodep  & \identity  \\
 \end{bmatrix}$,
as is was expected. 
\end{multicols}
\end{algorithm}
\end{example}

\begin{example}[Subtraction]\label{def:subtraction}
In this example we show how to type the subtraction between two variables.
\begin{algorithm}
 \begin{multicols}{2}
  \begin{algorithmic}
  \FUNOPEN{subtraction}{$S_1,S_2$}
    \LOOP{$S_2$}
      \State $\pop{S_1}$
    \ENDLOOP
  \FUNCLOSE{$S_1$}  \end{algorithmic}
  \vfil\columnbreak\vfill
  The function is typed with the identity matrix $\matrixID$, since the $\pop{}$ command is typed with the identity.
\end{multicols}
\end{algorithm}
\end{example}

\section{Semantics}

Semantics of the programs generated by the grammar in def \ref{def:gramdefinition} is the usual and expected one.
In the following we are using $\statefunc$ as the state function
associating to each variable a stack and to each register a letter.
Semantics for boolean value is labelled with probability, while semantics for expressions ($\artred$) is not carrying anything.
In figure \ref{def:semantics2} is shown the semantic
for booleans and expressions.
Most of boolean operator have probability $1$ and
operator rand reduced to $\true$ or $\false$ with probability $\frac{1}{2}$. 
Notice how there is no semantic associated to operators $op?()$ and $op()$. Of course, their semantics 
depends on how they will be implemented.

Since semantics for boolean is labelled with probability,
also semantics of commands ($\comred{\probone}$) is labelled with a probability, 
It tells us the probability to reach a particularly final state after having execute a command from a initial state.

\begin{figure*}[htbp]
\begin{center}
\fbox{
\begin{minipage}{.95\textwidth}
$$
\AxiomC{if $\statefunc(S) = \termstack{}$}
\UnaryInfC{$\pair{S}{\statefunc} \bolred{1} \true$}\DisplayProof
\hspace{5pt}
\AxiomC{if $\statefunc(S) != \termstack{}$}
\UnaryInfC{$\pair{S}{\statefunc} \bolred{1} \false$}\DisplayProof
\hspace{5pt}
\AxiomC{}
\UnaryInfC{$\pair{\rand}{\statefunc} \bolred{1/2} \true$}\DisplayProof
\hspace{5pt}
\AxiomC{}
\UnaryInfC{$\pair{\rand}{\statefunc} \bolred{1/2} \false$}\DisplayProof
$$
$$
\AxiomC{}
\UnaryInfC{$\pair{\regone}{\statefunc} \artred \statefunc(\regone)$}\DisplayProof
\hspace{5pt}
\AxiomC{if $\statefunc(S) = \termstack{c_1\ldots c_n}$}
\UnaryInfC{$\pair{\htop{S}}{\statefunc} \artred c_1$}\DisplayProof
$$
\end{minipage}
}
\end{center}
\caption{Semantics of booleans and expressions}
\label{def:semantics2}
\end{figure*}

In figure \ref{def:semantics} are presented the semantics for commands.
Since a compile time all the functions definitions can be collected, we suppose that exists a set of 
defined function called $\definedFunctions$ where all the functions defined belong.

\begin{figure*}[htbp]
\begin{center}
\fbox{
\begin{minipage}{.95\textwidth}\begin{scriptsize}
$$
\AxiomC{}\UnaryInfC{$\pair{\iskip}{\statefunc} \comred 1 \statefunc$}\DisplayProof
\hspace{5pt}
\AxiomC{$\statefunc(S)=\termstack{<c_1,c_2,\ldots,c_n>}$}
\UnaryInfC{$\pair{\pop{S}}{\statefunc} \comred{1} \subst{\statefunc}{S}{\termstack{c_2,\ldots,c_n}}$}\DisplayProof
\hspace{5pt}
\AxiomC{$\statefunc(S)=\termstack{<c_1,\ldots,c_n>}$}
\UnaryInfC{$\pair{\push{e,S}}{\statefunc} \comred{1} \subst{\statefunc}{S}{\termstack{e,c_1,\ldots,c_n}}$}\DisplayProof
$$
$$
\AxiomC{}\UnaryInfC{$\pair{\iassign{\regone}{\exprone}}{\statefunc} \comred{1} \subst{\statefunc}{\regone}{\exprone}$}\DisplayProof
\hspace{5pt}
\AxiomC{}\UnaryInfC{$\pair{\iassign{S_1}{S_2}}{\statefunc} \comred{1} \subst{\statefunc}{S_1}{\statefunc(S_2)}$}\DisplayProof
\hspace{5pt}
\AxiomC{}\UnaryInfC{$\pair{\iassign{S_1}{\termstack{c_1,\ldots,c_n}}}{\statefunc} \comred{1} \subst{\statefunc}{S_1}{\termstack{c_1,\ldots,c_n}} $}\DisplayProof
$$
$$
\AxiomC{$\deffunction{myfun}{S_1,\ldots,S_n}{C}{S_m} \in \definedFunctions$}
\AxiomC{$\pair{C}{\subst{\statefunc}{S_1,\ldots,S_n}{\multi S}} \comred{\probone} \statefuncone$}
\BinaryInfC{$\pair{\iassign{S_p}{ \call{myfun, \multi S}  }}{\statefunc} \comred{\probone} \statefuncone$}
\DisplayProof
$$
$$
\AxiomC{$\pair{\comone}{\statefunc_1}\comred\probone \statefunc_2$}
\AxiomC{$\pair{\comtwo}{\statefunc_2}\comred\probtwo \statefunc_3$}
\BinaryInfC{$\pair{\comone;\comtwo}{\statefunc} \comred{\probone\probtwo} \statefunc_3$}
\DisplayProof
\hspace{5pt}
\AxiomC{$\pair{\boolone}{\statefunc} \bolred\probone \true$}\AxiomC{$\pair{\comone}{\statefunc}\comred\probtwo \statefunc_1$}
\BinaryInfC{$\pair{\iifthenelse{\boolone}{\comone}{\comtwo}}{\statefunc} \comred{\probone\probtwo} \statefunc_1$}
\DisplayProof
$$
$$
\AxiomC{$\pair{\bigvar_\indexthree}{\statefunc} \artred \termlist{}$}
\UnaryInfC{$\pair{\iloop{\bigvar_\indexthree}{\comone}}{\statefunc}
  \comred 1 \statefunc$} \DisplayProof 
  \hspace{10pt}
\AxiomC{$\pair{\boolone}{\statefunc} \bolred\probone \false$}\AxiomC{$\pair{\comtwo}{\statefunc}\comred\probtwo \statefunc_1$}
\BinaryInfC{$\pair{\iifthenelse{\boolone}{\comone}{\comtwo}}{\statefunc} \comred{\probone\probtwo} \statefunc_1$}
\DisplayProof
$$
$$
\AxiomC{$\pair{\bigvar_\indexthree}{\statefunc} \artred \termstack{\smany{\constone}{n}{}}$}
\AxiomC{$\pair{\comone}{\statefunc} \comred{\probone_1} \statefunc_1$}
\noLine\UnaryInfC{$\ldots$}
\noLine\UnaryInfC{$\pair{\comone}{\statefunc_{\numeone-1}}
  \comred{\probone_\numeone} \statefunc_{\numeone}$} 
  \BinaryInfC{$ \pair{\iloop{\bigvar_\indexthree}{\comone}}{\statefunc}
  \comred{\Pi\probone_\indexone} \statefunc_{\numeone }$} \DisplayProof
$$
\end{scriptsize}
\end{minipage}
}
\end{center}
\caption{semantics of commands}
\label{def:semantics}
\end{figure*}

Since $\oursystem$ is working on stochastic computations, in order to reach
soundness and completeness respect to $\PP$, we need to define a
semantics for distribution of final states.  We need to introduce some
more definitions. Let $\distrone$ be a distribution of probabilities
over states. Formally, $\distrone$ is a function whose type is
$(\variable\rightarrow\values)\rightarrow\probone$. Sometimes we will
use the following notation $ \distrone = \{\statefunc_1^{\probone_1},
\ldots, \statefunc_\numeone^{\probone_\numeone} \}$ indicating that
probability of $\statefunc_\indexone$ is $\probone_\indexone$.

We can so define semantics for distribution; the most important rules are 
shown in Figure~\ref{def:semanticdistributionrules}. Since semantics for some commands computes with probability equal to $1$,
the correspondent rule for distributions is not presented.
Unions of distributions and multiplication between real number
and a distribution have the natural meaning.
Notice also how all the final distributions are normalized distributions.

\begin{figure}[htbp]
\begin{center}
\fbox{
\begin{minipage}{.95\textwidth}\scriptsize
$$
\AxiomC{$\pair{\comone}{\statefunc}\comdistr \distrone$}
\AxiomC{$\forall \statefunc_\indexone \in
  \distrone. \pair{\comtwo}{\statefunc_\indexone}\comdistr
  \distrtwo_\indexone$} \BinaryInfC{$\pair{\comone;\comtwo}{\statefunc}
  \comdistr
  \bigcup_\indexone\distrone(\statefunc_\indexone)\cdot\distrtwo_\indexone$}
\DisplayProof
$$
$$
\AxiomC{$\pair{\bigvar_\indexthree}{\sigma} \artred 0$}
\UnaryInfC{$\pair{\iloop{\bigvar_\indexthree}{\comzero}}{\sigma}
  \comdistr \{\sigma^1\}$} \DisplayProof 
  \hspace{10pt}
\AxiomC{$\pair{\bigvar_\indexthree}{\sigma} \artred \termstack{\smany{\constone}{n}{}}$}
\AxiomC{$\pair{\,\overbrace{\comzero;\comzero;\ldots;\comzero}^{\numeone}}{\sigma}
  \comdistr \distrtwo$} 
  \BinaryInfC{$\pair{\iloop{\bigvar_\indexthree}{\comzero}}{\statefunc} \comdistr
  \distrtwo$} \DisplayProof
$$
$$
\AxiomC{$\pair{\boolzero}{\sigma} \bolred\probone
  \true$}\AxiomC{$\pair{\comone}{\sigma}\comdistr \distrone$}
\AxiomC{$\pair{\comtwo}{\sigma}\comdistr \distrtwo$}
\TrinaryInfC{$\pair{\iifthenelse{\boolzero}{\comone}{\comtwo}}{\sigma}
  \comdistr (\probone\cdot\distrone) \cup ((1-\probone)\cdot\distrtwo)$}
\DisplayProof
$$
\end{minipage}
}
\end{center}
\caption{Distributions of output states}\label{def:semanticdistributionrules}
\end{figure}

Here we can present our first result. 
\begin{theorem}
  A command $\comzero$ in a state $\statefuncone$ reduce to another
  state $\statefunctwo$ with probability equal to
  $\distrone(\statefunctwo)$, where $\distrone$ is the distribution of
  probabilities over states such that $\pair{\comzero}{\statefuncone}
  \comdistr \distrone$.
\end{theorem}

Proof is done by structural induction on derivation tree. It is quite
easy to check that this property holds, as the rules in
Figure~\ref{def:semanticdistributionrules} are showing us exactly this
statement. The reader should also not be surprised by this
property. Indeed, we are not considering just one possible derivation
from $\pair{\comone}{\statefuncone}$ to $\statefunctwo$, but all the
ones going from the first to the latter.

\section{Soundness}

The language recognised by $\oursystem$ is an imperative language where
the iteration schemata is restricted and the size of objects (here,
stacks) is bounded. These are ingredients of a lot of well known ICC
polytime systems. There is no surprise that every program certified by
$\oursystem$ runs in probabilistic polytime. 

Now we can start to present theorems and lemmas of our system. First
 we will focus on multipolynomial properties in order to show that the
 behaviour of these algebraic constructor is similar to the behaviour of
 matrices in our system. Finally we will link these things together to
 get polytime bound for $\oursystem$.  Here are two fundamental lemmas. Their proofs are straightforward.
 
 \begin{lemma}\label{lm:polysum}
 Let $\polyone$ and $\polytwo$ two positive polynomials, then it holds that $\polabstr{\polyone\oplus\polytwo} 
 = \polabstr{\polyone} \cup \polabstr{\polytwo}$.
 \end{lemma}
 \begin{proof} by induction on the size of the two polynomials.
  By definition the union between two polynomial is defined in \ref{def:unionmultipolinomial} as the maximum of the comparable
  monomials. Let's analyze the different cases:
  \begin{itemize}
   \item If $\polyone = c_1 + \polythree$ and $\polytwo = c_2 + \polyfour$. By induction, 
   $\polabstr{\polythree\oplus\polyfour} = \polabstr{\polythree} \cup \polabstr{\polyfour}$. 
   By definition \ref{def:unionmultipolinomial}, $\polyone \oplus \polytwo$ is $\max{(\alpha,\beta)} + \polyunion{\polythree}{\polyfour}$
   and so, by definition \ref{def:abstractionmultipolynomials}, the abstraction is defined as $\polabstr{\max{(\alpha,\beta)}} +
   \polabstr{\polyunion{\polythree}{\polyfour}}$. By using induction hypothesis we get 
   $\polabstr{\max{(\alpha,\beta)}} + \polabstr{\polythree} \cup \polabstr{\polyfour}$. It's clear that $\polabstr{\max{(\alpha,\beta)}}$
   is equal to $\max{(\polabstr{\alpha},\polabstr{\beta})}$, since the abstraction take in account the value of the constants. 
   This is, by definition, the union of the two abstracted polynomials. We get, so 
   $(\polabstr{\alpha}\cup \polabstr{\beta}) + (\polabstr{\polythree} \cup \polabstr{\polyfour})$.
   Notice how the abstractions 
   of the two constants are two column vectors having $\nodep$ everywhere except for the last row.
   We can so rewrite the previous equation as 
   $(\polabstr{\alpha} +\polabstr{\polythree}) \cup (\polabstr{\beta} +\polabstr{\polyfour})$, that is the thesis.
   \item If $\polynormform\polyone = \alpha X_i+\polythree$ and $\polynormform\polytwo = \beta X_i + \polyfour$. This case is very similar
   to the previous one.
   \item The last case is where the two polynomials are not comparable. In this case, the union is defined as $\polyone + \polytwo$. 
   There are two cases:
	\begin{itemize}
	 \item If some variables are present just in one polynomial and not in the other one, then the 
	 correspondent rows, for each single variable, is not influenced by the abstraction of the polynomial in which the variable does
	 not appear.
	 \item If some variables are present in both. In this case it means that the variables appear in at least one monomial 
	 with grade gretar than one or in a monomial having more than one variable. In both cases the associated abstracted value for both 
	 is $\multiplied$.
	\end{itemize}
  The thesis holds. This concludes the proof.
  \end{itemize}
 \end{proof}

 \begin{lemma}\label{lm:multipolsum} Let $\multipolyone$ and $\multipolytwo$ two
   positive multipolynomials, then it holds that
   $\polabstr{\polyunion{\multipolyone}{\multipolytwo}} = \polabstr{\multipolyone}
   \cup \polabstr{\multipolytwo}$.
 \end{lemma}
 \begin{proof}
 By definition of sum between multipolynomials \ref{def:unionmultipolinomial} we know that sum is defined componentwise,
  $(\multipolyone \oplus \multipolytwo)_\indexone =  \polyunion{\multipolyone_\indexone}{\multipolytwo_\indexone} $.
  By lemma \ref{lm:polysum} we prove the theorem.
 \end{proof}

 \begin{lemma}\label{lm:multipolmult}
   Let $\multipolyone$ and $\multipolytwo$ two positive multipolynomials (over $n$ variables) in
   canonical form, then it holds that $\polabstr{\multipolyone\cdot
   \multipolytwo} \le \polabstr{\multipolytwo}\times \polabstr{\multipolyone}$
 \end{lemma}
 \begin{proof}
 We will consider the element in position $\indexone,\indextwo$ and so we have:
  $(\polabstr{\multipolytwo}\times \polabstr{\multipolyone})_{\indexone,\indextwo} = 
  \sum_\indexthree{ \polabstr{\multipolytwo}_{\indexone,\indexthree} \times  \polabstr{\multipolyone}_{\indexthree,\indextwo}} $.
  We can start by making some algebraic passages:
  \[
    \polabstr{\multipolyone\cdot\multipolytwo}_{\indexone,\indextwo} = 
    \polabstr{ \multipolyone(\multipolytwo_1,\ldots,\multipolytwo_{n})}_{\indexone,\indextwo} =
    \polabstr{ \multipolyone_\indextwo(\multipolytwo_1,\ldots,\multipolytwo_{n})}_{\indexone}
  \]
  
  The equality holds because we are considering the element in the $\indextwo$-th column. Since we are interested at the element in 
  position $\indexone$-th we have to understand how the variable $X_i$ (or constant) in each $\multipolytwo_k$ is substituted. 
  
  \begin{itemize}
  \item Case where $\indexone = \numeone+1$. 
      \begin{itemize}
       \item If none of the polynomials $\multipolytwo_k$ has a constant inside, then the proof is evident, since the only possible constant appearing
	    in the result is the possible constant appearing in $\multipolyone_j$. 
	    Recall that the element in position $(\numeone+1,\numeone+1)$ is $\identity$ by definition, so 
	    $\sum_\indexthree{ \polabstr{\multipolytwo}_{\indexone,\indexthree} \times  \polabstr{\multipolyone}_{\indexthree,\indextwo}} $
	    contain at least $\polabstr{\multipolyone}_{\numeone+1,\indexone}$.
       \item Otherwise some constants appear in some $\multipolytwo_k$. This means that the expected abstraction for the element at position
	    $(\numeone+1,\indextwo)$ may be $\added$ or $\identity$. If $\identity$ is the result, then is clear that and equality holds,
	    since it means that the constant is $1$.
	    The inequality hold if the expected result is $\added$, since that on the right side we have to 
	    perform the following sum:
	    $\sum_\indexthree{ \polabstr{\multipolytwo}_{\indexone,\indexthree} \times  \polabstr{\multipolyone}_{\indexthree,\indextwo}}$
	    and we could find an $\added$ or $\multiplied$ value.
      \end{itemize}
  \item Case where $\indexone < \numeone +1$. In this case we are considering how the variable $X_i$ appears.
	We have four possibilities:
	    \begin{itemize}
	     \item If $X_\indexone$ does not appear in any $\multipolytwo_k$ polynomials. In this case the expected abstract value is $\nodep$.
		    Is easy to check that this holds, since on the left side of the inequality we get $\nodep$ and on the right side we get
		    $\sum_\indexthree{ \polabstr{\multipolytwo}_{\indexone,\indexthree} \times  \polabstr{\multipolyone}_{\indexthree,\indextwo}}$
		    that is $\nodep$, since all $\polabstr{\multipolytwo}_{\indexone,\indexthree}$ are $\nodep$.
	     \item In the following we will consider that $X_\indexone$ appears in some $\multipolytwo_\indexthree$ polynomials.
		   Call them $\multi\multipolytwo_{X_i}$.
		   If some of the polynomials where $X_\indexone$ appears is substituted in some monomial of $\multipolyone_j$
		   of shape as $\alpha X_\indexfour\cdot\polytwo(\multi X)$ in place of some $X_\indexfour$, then for sure on the right 
		   side of the inequality we will get a value $\multiplied$. On the right side, considering 
		   $\sum_\indexthree{ \polabstr{\multipolytwo}_{\indexone,\indexthree} \times  \polabstr{\multipolyone}_{\indexthree,\indextwo}}$
		   we will multiply for sure an $\multiplied$ value with the abstracted value for $X_\indexone$ of the 
		   $\multi\multipolytwo_{X_i}$ where it appears. The result is so for sure $\multiplied$.
	     \item Otherwise, if some of the polynomials where $X_\indexone$ appears is substituted in some monomial of $\multipolyone_j$
		   of shape as $\alpha X_\indexfour$ ($\alpha > 1$), then the expected abstract value depends on how $X_i$ appears in
		   $\multi\multipolytwo_{X_i}$. For all the three possible cases of $X_i$ in $\multi\multipolytwo_{X_i}$
		   the abstracted value obtained on the left side is equal to the value obtained on the right side.
	     \item Otherwise, 
		   $X_\indexone$ appears is substituted in some monomial of $\multipolyone_j$
		   of shape as $X_\indexfour$; then the substitution gives in output exactly the
		   $\multi\multipolytwo_{X_i}$ substituted. The equality holds because on the right side we are going to multiply
		   by $\identity$ the abstracted value found for each $\multipolytwo_k$.
	    \end{itemize}
   This concludes the proof.
  \end{itemize}
 \end{proof}

 Let's now present the results about the probabilistic polytime soundness. 
 The following theorem tell us that at each step of execution of a program, size of variables are polynomially 
 correlated with size of variables in input.

\begin{theorem}\label{th:polysizeboundcommands}
 Given a command $\comzero$ well typed in $\oursystem$ with matrix $\matrixone$, such that 
 $\pair{\comzero}{\statefuncone}\comred{\alpha} \statefunctwo$ we get that exists a multipolynomial 
 $\multipolyone$ such that for all stacks $\bigvar_\indexone$ we have that 
 $\size{\statefunctwo(\bigvar_\indexone)}\le \multipolyone_{\indexone}(\size{\statefuncone(\bigvar_1)},
 \ldots,\size{\statefuncone(\bigvar_\numeone)}  )$ and $\polabstr{\multipolyone}$ is $\matrixone$.
\end{theorem}
\begin{proof}
 By structural induction on typing tree. We will present just the most important cases.
 \begin{itemize}
  \item If the last rule is \textsc{(Const-0)}, it means that we have only one stack and its size is $0$. The relative vector in the matrix
  is a $\vectorO$. We can choose the constant polynomial $0$, whose abstraction is exactly $\vectorO$. The polynomial $0$ bounds the size
  of the stack.
  \item If the last rule is one of the following \textsc{(Skip)}, \textsc{(Const-A)}, \textsc{(Const-L)}, 
  \textsc{(Axiom-Reg)}, then the proof is trivial.
  \item If the last rule is \textsc{(Push)}, then we know that the size of of the stack $S_i$ has been increased by $1$. 
  The associated vector is a column vector having $\identity$ on the $i$-th row and $\identity$ on the last line.
  The correspondent polynomial, $X_i + 1$, is the correct bounding polynomial for the $i$-th stack.
  \item If the last rule is \textsc{(Subtyp)}, then by induction on the hypothesis we can easily find a new polynomial bound.
  \item If the last rule is \textsc{(Asgn)}, then we know that the size of the $i$-th stack is equal to the size of the stack $j$-th.
  So, the polynomial bounding the size of the $i$-th stack uses at least two variables and the correct one is $P_i(X_i,X_j)=X_j$.
  \item If the last rule is \textsc{(IfThen)}, then by applying induction hypothesis on the two premises we have multipolynomial bounds $\multipolytwo$, $\multipolythree$ such that $\polabstr{\multipolytwo} = \matrixone$ and $\polabstr{\multipolythree} = \matrixtwo$.
  By lemma \ref{lm:multipolsum} we get the thesis.
  \item If the last rule is \textsc{(Fun)}, then by applying the induction hypothesis on the premise we directly prove the thesis.
  \item If the last rule is \textsc{(FunCall)}, then by applying the induction hypothesis on the premise we have a polynomial bound $\multipolytwo$ such that $\polabstr{\multipolytwo} = \matrixone$. For all the stacks different from the $i$-th, the bounding polynomial
  is trivial, while for the stack $S_i$ depends on the result of the function call.
  
  The function return the value stored in the $j$-th stack used inside the function. According to the actual parameters, the actual 
  polynomial bound is different from the one retrieving by applying the induction hypothesis.
  
  \item If last rule is \textsc{(Concat)}, then by lemma \ref{lm:multipolmult} we can easily conclude the thesis.
  \item If last rule is \textsc{(Loop)}, we are in the following case; so, $\matrixone$ is 
  $\mergedown{ (\closure{\matrixtwo}) }{\indexthree}$. The typing and the associated semantic are the following:
$$  
\AxiomC{$
  \vdash \comone : \matrixtwo $} \AxiomC{$ \forall \indexone,
  {(\closure\matrixtwo)}_{\indexone,\indexone} < \added $}
\RightLabel{\textsc{(Loop)}} \BinaryInfC{$ \vdash
  \iloop{\bigvar_\indexthree}{\comone} :
  \mergedown{(\closure{\matrixtwo})}{\indexthree}$} \DisplayProof
\hspace{10pt}
 \hspace{10pt}
\AxiomC{$\pair{\bigvar_\indexthree}{\statefunc} \artred \termstack{\smany{\constone}{n}{}}$}
\AxiomC{$\pair{\comone}{\statefunc} \comred{\probone_1} \statefunc_1$}
\noLine\UnaryInfC{$\ldots$}
\noLine\UnaryInfC{$\pair{\comone}{\statefunc_{\numeone-1}}
  \comred{\probone_\numeone} \statefunc_{\numeone}$} 
  \BinaryInfC{$ \pair{\iloop{\bigvar_\indexthree}{\comone}}{\statefunc}
  \comred{\Pi\probone_\indexone} \statefunc_{\numeone }$} \DisplayProof
$$

We consider just the case where $n>0$, since the other one is trivial.
By induction on the premise we have a multipolynomial $\multipolyone$ bound for command $\comone$ such that its abstraction is $\matrixtwo$. 
 If $\multipolyone$ is a bound for $\comone$, then $\multipolyone\cdot\multipolyone$ is a bound for $\comone;\comone$ and $(\multipolyone\cdot\multipolyone)\cdot\multipolyone$ is a bound for $\comone;\comone;\comone$ and so on. All of these are multipolynomial because we are composing multipolynomials with multipolynomials.
 
 By lemma \ref{lm:multipolmult} and knowing that $\polabstr{\multipolyone}$ is $\matrixtwo$ we can easily deduce to have a multipolynomial bound for every iteration of command $\comone$. In particularly by lemma \ref{lm:multipolsum} we can easily sum up everything and find out a multipolynomial $\multipolytwo$ such that $\polabstr{\multipolytwo}$ is $\closure{\matrixtwo}$.
 This means that further iterations of sum of powers of $\multipolyone$ will not change the abstraction of the result.

 So, for every iteration of command $\comone$ we have a multipolynomial bound whose abstraction cannot be greater than $\closure{\matrixtwo}$. So, we study the worst case; we analyse the matrix $\closure\matrixtwo$.
 
 Side condition on \textsc{(Loop)} rule tells us to check elements on the main diagonal. Recall that by definition of union closure, elements on the main diagonal are supposed to be greater then $\nodep$. We required also to be less then $\added$.
 Let's analyse all the possibilities of an element in position $\indexone,\indexone$:
  \begin{itemize}
  \item Value $\nodep$ means no dependencies. If value is $\identity$ it means that $\multipolytwo_\indexone$ concrete bound for such column has shape $\bigvar_\indexone + \polythree(\multi\bigvar)$, where $\bigvar_\indexone$ does not appear in $\polythree(\multi\bigvar)$. Iteration of such assignment gives us polynomial bound increment of the value of variable $\bigvar_\indexone$.
  \item If value is $\added$ could means that $\multipolytwo_\indexone$ concrete bound for such column has shape $\alpha\bigvar_\indexone + \polythree(\multi\bigvar)$ (for some $\alpha>1$), where $\bigvar_\indexone$ does not appear in $\polythree(\multi\bigvar)$. Iteration of such assignment lead us to exponential blow up on the size of $\bigvar_\indexone$.
  \item Otherwise value is $\multiplied$. This case is worse than the previous one. It's evident that we could have exponential blow up on the size of $\bigvar_\indexone$.
 \end{itemize}
 
 The abstract bound ${\matrixtwo}$ is still not a correct abstract bound for the loop because loop iteration depends on some variable $\bigvar_\indexthree$. We need to adjust our bound in order to keep track of the influence of variable $\bigvar_\indexthree$ on loop iteration. 
 
 We take multipolynomial $\multipolytwo$ because we know that further iterations of the algorithm explained before will not change its abstraction $\polabstr{\multipolytwo}$. Looking at $\indexone$-th polynomial of multipolynomial $\multipolytwo$ we could have three different cases. We behave in the following way:
 \begin{itemize}
  \item The polynomial has shape $\bigvar_\indexone + \polyone(\multi\bigvar)$. In this case we multiply the polynomial $\polyone$ by $\bigvar_\indexthree$ because this is the result of iteration. We substitute the $\indexone$-th polynomial with the canonical form of polynomial $\bigvar_\indexone + \polyone(\multi\bigvar)\cdot\bigvar_\indexthree$.
  \item The polynomial has shape $\bigvar_\indexone + \alpha$, for some constant $\alpha$. In this case we substitute with $\bigvar_\indexone + \alpha\cdot\bigvar_\indexthree$.
  \item The polynomial has shape $\bigvar_\indexone$ or $\bigvar_\indexone$ does not appear in the polynomial. We leave as is.
 \end{itemize}
 
 In this way we generate a new multipolynomial, call it $\multipolythree$.
 The reader should easily check that these new multipolynomial expresses a good bound of iterating $\multipolytwo$ a number of times equal to $\bigvar_\indexthree$. Should also be quite easy to check that $\polabstr{\multipolythree}$ is exactly $\mergedown{(\closure{\matrixtwo})}{\indexthree}$. This concludes the proof.
\end{itemize}
\end{proof}

Polynomial bound on size of stacks is not enough; we should also prove polynomiality of number of steps. Since all the programs generated by the language terminate and all the stacks are polynomially bounded in their size, the theorem follows straightforward.

\begin{theorem}\label{th:polytimeboundcommands}
 Let $\comzero$ be a command well typed in $\oursystem$ and $\statefuncone ,\statefunc_\numeone$ state functions. If $\derone:\pair{\comone}{\statefunc_0}\comred{\alpha}\statefunc_\numeone$, then there is a polynomial $\polyone$ such that $\size{\derone}$ is bounded by $\polyone(\sum_\indexone \size{\statefunc_0(\bigvar_\indexone)  } )$.
\end{theorem}
\begin{proof}
 By induction on the associated semantic proof tree.
\end{proof}

\subsection{Probabilistic Polynomial Soundness}

Nothing has been said about probabilistic polynomial soundness. Theorems \ref{th:polysizeboundcommands} and \ref{th:polytimeboundcommands} tell us just about polytime soundness. Probabilistic part is now introduced. We will prove probabilistic polynomial soundness following idea in \cite{PPTUDL2011}, by using ``representability by majority''.

\begin{definition}[Representability by majority]\label{def:rappresentbymajority}
 Let $\subst{\multi\sigmazero}{\bigvar}{\numeone}$ define as $\forall \bigvar, \sigmazero(\bigvar)=\numeone$.
Then $\comzero$ is said to \emph{represent-by-majority} a language
$L\subseteq\NN$ iff:
\begin{enumerate}
\item
  If $n\in L$ and $\pair{\comzero}{ \subst{\multi\sigmazero}{\bigvar}{\numeone} } \comdistr \distrone$, then $\distrone(\sigma_0) \ge \sum_{m>0}\distrone(\sigma_\numetwo)$;
\item
  If $n\notin L$ and $\pair{\comzero}{ \subst{\multi\sigmazero}{\bigvar}{\numeone} } \comdistr \distrone$, then $\sum_{\numetwo>0}\distrone(\sigma_\numetwo) > \distrone(\sigma_0)$.
\end{enumerate} 
\end{definition}

Observe that every command $\comzero$ in $\oursystem$ represents by majority a language as defined in \ref{def:rappresentbymajority}. In literature \cite{AroraBarak} is well known that we can define $\PP$ by majority. We say that the probability error should be at most $\frac{1}{2}$ when we are considering string in the language and strictly smaller than $\frac{1}{2}$ when the string is not in the language.
So we can easily conclude that $\oursystem$ is sound also respect to probabilistic polytime.

\section{Probabilistic Polynomial Completeness}

There are several way to demonstrate completeness respect to some complexity class. We will show that by using language recognised by our
system we are able to encode Probabilistic Turing Machines (PTM). We will are not able to encode all possible PTMs but all the ones with 
particularly shape. This lead us to reach extensional completeness. For every problem in $\PP$ there is at least an algorithm solving that 
problem that is recognised by $\oursystem$.

A Probabilistic Turing Machine \cite{Gill77} can be seen as non deterministic TM with one tape where at each iteration are able to flip a coin and choose
between two possible transition functions to apply.

In order to encode Probabilistic Turing Machines we will proceed with the following steps:
\begin{itemize}
 \item We show that we are able to encode polynomials. In this way we are able to encode the polynomial representing the number of steps required by the machine to complete.
 \item We encode the input tape of the machine.
 \item We show how to encode the transition $\delta$ function.
 \item We put all together and we have an encoding of a PTM running in polytime.
\end{itemize}

Should be quite obvious that we can encode polynomials in $\oursystem$. Grammar and examples \ref{def:addition}, \ref{def:multiplicationconstant}, \ref{def:multiplication}, \ref{def:subtraction} give us how encode polynomials.

We need to encode the tape of our PTMs. We subdivide our tape in three sub-tapes. The left part $\tapel$, the head $\tapeh$ and the right part $\taper$. $\taper$ is encoded right to left, while the left part is encoded as usual left to right. 

Let's move on and present the encoding of transition function of PTMs. Transition function of PTMs, denoted with $\delta$, is a relation $\delta \subseteq (Q\times\Sigma)\times(Q\times\Sigma\times\{\leftarrow,\downarrow,\rightarrow\})$. Given an input state and a symbol it may give in output more tuples of state, a symbol and a direction of the head (left, no movement, right).

In the following we are going to present two procedures to encode movements of the head. It is really important to pay attention on how we encode this operations. Recall that a PTM loops the $\delta$ function and our system requires that the matrix certifying/typing the loop needs to have values of the diagonal less than $\added$. 

\begin{definition}[Move head to right]
Moving head to right means to concatenate the bit pointed by the head to the left part of the tape; therefore we need to retrieve the first bit of the right part of the tape and associate it to the head. Procedure is presented as algorithm \ref{alg:movetoright}; call it \textsc{MoveToRight()}.
\end{definition}

\begin{algorithm}
\begin{multicols}{2}
 \begin{algorithmic}
 \State
\State $\push{\htop{\tapeh},\tapel}$
\State $\iassign{\tapeh}{\termlist{}}$
\State $\push{\htop{\taper},\tapeh}$
\State $\pop{\taper}$
\end{algorithmic}
\vfill\columnbreak\vfill
Using typing rules we are able to type the algorithm with the following matrix:
$$\begin{bmatrix}
\identity & \nodep    & \nodep    & \nodep     & \nodep\\
\nodep    & \nodep    & \nodep    & \nodep     & \nodep\\
\nodep    & \nodep    & \identity & \nodep     & \nodep\\
\nodep    & \nodep    & \nodep    & \identity  &\nodep\\
\identity & \added    & \nodep    &\nodep      & \identity  \\
\end{bmatrix}$$
\end{multicols}
\caption{Move head to right}\label{alg:movetoright}
\end{algorithm}

The first column of the matrix represents dependencies for variables $\tapel$, the second represents $\tapeh$, third is $\taper$, forth is $\statemachine$ and finally recall that last column is for constants. In the following, columns of matrices are ordered in this way.

Similarly we can encode the procedure for moving the head to left and the possibility of not moving at all, that is a $\iskip$ command. So, the $\delta$ function is then encoded in the standard way by having nested \texttt{If-Then-Else} commands, checking the value of $\rand$, the state, the symbol on a tape and performing the right procedure.

\begin{algorithm}
\begin{multicols}{2}
\begin{algorithmic}
\State
\If {$\rand$}
  \If {$\test{\statemachine,1}$}
  \Else
    \If {$\test{\statemachine,2}$}
      \State $\cdots$
    \Else
      \State $\cdots$
    \EndIf
  \EndIf
\Else
  \State $\cdots$
\EndIf
\end{algorithmic}
\vfill\columnbreak\vfill
The prototype is created by nesting \texttt{If-Then-Else} commands and checking the state of the machine. 
for each branch, then, an operation of moving the head is performed. Notice that since three possible 
operations could be performed, all the nested \texttt{If-Then-Else} are typed with the following matrix:
$$\begin{bmatrix}
\identity & \nodep    & \nodep    & \nodep     & \nodep\\
\nodep    & \nodep    & \nodep    & \nodep     & \nodep\\
\nodep    & \nodep    & \identity & \nodep     & \nodep\\
\nodep    & \nodep    & \nodep    & \identity  &\nodep\\
\identity & \added    & \identity &\nodep      & \identity  \\
\end{bmatrix}$$
\end{multicols}
\caption{Prototype of encoded $\delta$ function}\label{alg:deltafunction}
\end{algorithm}

Finally, we have to put the encoded $\delta$-function inside a loop. 
The machine runs in a polynomial number of steps. Since the encoded $\delta$-function is typed with the matrix 
presented in Alg. \ref{alg:deltafunction}, we can easily see that the union closure of that matrix fits the constraints
of the typing rule of $\iloop{}{}$. We can therefore conclude that we can encode Probabilistic Turing Machine working in polytime.

\section{Polynomiality}

In this last session we will discuss why $\oursystem$ is a feasible analyser. We already shown that is sound, 
so is able to understand whenever a program does not run in Probabilistic Polynomial Time. Moreover, is also complete, respect to $\PP$;
this means that $\oursystem$ is able to recognise a lot of programs. At least one for each problem in $\PP$. The final question has to 
do with the efficiency of our system: ``how much time does it take $\oursystem$ to check a program?''.
Can be shown that $\oursystem$ is running in polytime respect to the number of variables used.

Since the typing rules are deterministic, the key problems lays on the rule \textsc{(Loop)}.

$$
\AxiomC{$
  \vdash \comone : \matrixone $} \AxiomC{$ \forall \indexone,
  {(\closure\matrixone)}_{\indexone,\indexone} < \added $}
\RightLabel{\textsc{(Loop)}} \BinaryInfC{$ \vdash
  \iloop{\bigvar_\indexthree}{\comone} :
  \mergedown{(\closure{\matrixone})}{\indexthree}$} \DisplayProof
$$

It is not trivial to understand how much it takes a union closure to be performed. While all the typing rules for all the
other commands and expressions are trivial, the one for loop needs some more explanations.
By definition, $\closure{\matrixone}$ is defined as $\cup_{\indexone} \matrixone^\indexone$. 

\begin{figure}\fbox{
\begin{minipage}{.95\textwidth}
Every matrix could be seen as adjacency matrix of a graph. 
\begin{multicols}{2}
As example, the following matrix $\matrixone$:
$$\begin{bmatrix}
 \identity    & \nodep         &\nodep      &\nodep\\
 \identity    & \identity      &\identity   &\nodep\\
 \identity    & \nodep         &\nodep      &\nodep\\
 \nodep       & \multiplied    &\identity   &\identity\\
\end{bmatrix}$$
has its own representation in the graph on the right side.
\begin{center}
 \includegraphics[scale=0.7]{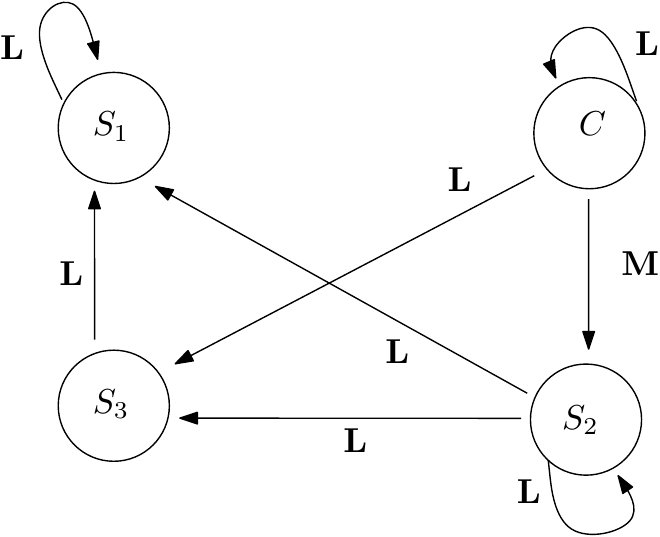}
\end{center}
\end{multicols}
\caption{Example of graph representing a matrix in $\oursystem$}\label{fig:examplegraph1}
\end{minipage}}
\end{figure}

In the example in Figure \ref{fig:examplegraph1} we can easily check that 
$C$ flows in $S_1$ with $\multiplied$ in one step. So, $\matrixone^2$ have $\multiplied$ in position $(4,1)$, 
$(4,2)$, $(4,3)$. Indeed, by using the rule of our algebra we can see how dependencies flows in the graph.

How many unions have to be performed in order to calculate $\cup_{\indexone} \matrixone^\indexone$?
In order to answer to this question, we can prove the following theorem.

\begin{theorem}[Polynomiality]
 Given a squared matrix $\matrixone$ of size $\dimensionone$ and $\matrixtwo=\bigcup_i{\matrixone^\indexone}$, 
 we get that $\matrixtwo = \bigcup_{\indexone < \dimensionone^2}\matrixone^\indexone$. Union closure can be calculated by considering
 just the first $\dimensionone^2$ matrix power.
\end{theorem}
\begin{proof}
 Here is the scratch of the proof. Since the matrix is an encoding of a flow graph, we can see the matrix as a graph of dependencies
 between stacks size. Recall that the union is component-wise, so we can focus on a singular element of a matrix. Given two nodes $S_1$ and
 $S_2$ of our graph, let's check all the possibilities:
 \begin{itemize}
  \item The expected value is $\multiplied$. If so, after no more than $\numeone$ iteration of $\matrixone$ we should have found it.
  If not, there are no possibilities to have $\multiplied$ in that position. After $\numeone$ iteration, the information has flown 
  through all the nodes.
  \item The expected value is $\added$. We need to iterate more than $\numeone$ times. Indeed $\added$ value can be found also by 
  adding $\identity+\identity$. In the flow-graph relation, this means finding two distinct paths from node $S_1$ to $S_2$. This can be
  easily done by encoding two paths in one. By generating all the possible pairs of nodes, we can easily see that the number of steps
  to find, if exists, two distinct paths takes $\dimensionone^2$ number of steps (number of all pairs).
  \item If after $\dimensionone^2$ steps no $\multiplied$ or $\added$ value has been found, the maximum value found is the correct one. 
  Indeed, if no dependence has been found or if just a linear dependence has been found, no further iteration could change the final
  value.
 \end{itemize}
\end{proof}

\section{Conclusions}

We presented an ICC system characterising the class $\PP$. There are several improvements respect to the known systems in literature.
We can catalogue them in two sets. First, we extend the known system to probabilistic computations, being able to characterise $\PP$. 
Since the typing requires polynomial time, it is feasible to use $\oursystem$ as a static analyser for complexity. 
The typing/certificate gives also information about the polynomial bound.
On the other hand, respect to sequential computations, we presented a finer analysis. $\oursystem$ works over a concrete language
and takes care of constants and function calls.
For all of these reasons, we are able to show a program that cannot be typed correctly by Kristiansen and Jones \cite{Jones2009a}.

\begin{algorithm}
\begin{multicols}{2}
 \begin{algorithmic}
  \LOOP{$S_1$}
    \State $\iassign{S_2}{0 * S_2}$
  \ENDLOOP
 \end{algorithmic}
 \vfill\columnbreak\vfill
 That is typed with the identity matrix $\matrixsubstitution{\matrixID}{2}{\vectorO}$. For multiplication we use the implementation in def \ref{def:multiplicationconstant}.
 \end{multicols}
  \label{alg:recognizedprogram}\caption{Example of recognised program}
\end{algorithm}

Since every constant is abstracted as a variable in \cite{Jones2009a}, they cannot for sure recognise that this program runs in polytime 
and for this reason this program should be rejected. Once abstracted it is impossible to know the value of the constant.
Of course, everything depends on how the abstraction is made.
In general, every program which deals with constants could appear problematic 
in \cite{Jones2009a} \cite{BenAmram2008}; At least, for a lot of programs, their bounds are bigger.
Moreover, as they wrote in \cite{BenAmram2008}: ``Note that no procedure
for inferring complexity will be complete for $L_{concrete}$'', while our procedure is sound and complete for our concrete language.

Finally we would like to point out some future direction:
\begin{itemize}
 \item Integrating the analysis with new features in order to capture more programs.
 \item Apply our analysis to a more generic imperative programming language.
 \item Extending the algebra in such way that the associated certificates would tell more detailed
information about the polynomial bounding the complexity.
 \item Make a finer analysis in order to be sound and complete for $\BPP$.
\end{itemize}

\bibliographystyle{plain}
\bibliography{bibliography}

\begin{thebibliography}{10}

\bibitem{AroraBarak}
Sanjeev Arora and Boaz Barak.
\newblock {\em Computational Complexity, A Modern Approach}.
\newblock Cambridge University Press, 2009.

\bibitem{Bellantoni1992}
Stephen Bellantoni and Stephen~A. Cook.
\newblock A new recursion-theoretic characterization of the polytime functions.
\newblock {\em Computational Complexity}, 2:97--110, 1992.

\bibitem{BenAmram2008}
Amir~M. Ben-Amram, Neil~D. Jones, and Lars Kristiansen.
\newblock Linear, polynomial or exponential? complexity inference in polynomial
  time.
\newblock In {\em Proceedings of the 4th conference on Computability in Europe:
  Logic and Theory of Algorithms}, CiE '08, pages 67--76, Berlin, Heidelberg,
  2008. Springer-Verlag.

\bibitem{BMMTCS}
G.~Bonfante, J.-Y. Marion, and J.-Y. Moyen.
\newblock Quasi-interpretations a way to control resources.
\newblock {\em Theoretical Computer Science}, 412(25):2776 -- 2796, 2011.

\bibitem{Cob65}
Alan Cobham.
\newblock The intrinsic computational difficulty of functions.
\newblock In Y.~Bar-Hillel, editor, {\em Logic, Methodology and Philosophy of
  Science, proceedings of the second International Congress, held in Jerusalem,
  1964}, Amsterdam, 1965. North-Holland.

\bibitem{PPTUDL2011}
Ugo Dal~Lago and Paolo Parisen~Toldin.
\newblock A higher-order characterization of probabilistic polynomial time.
\newblock In R.~Pe\~na, M.~van Eekelen, and O.~Shkaravska, editors, {\em
  {Proceedings of $2^{nd}$ International Workshop on Foundational and Practical
  Aspects of Resource Analysis, FOPARA 2011}}, volume 7177 of {\em LNCS}.
  Springer, 2011.
\newblock To be appeared in.

\bibitem{Gill77}
John Gill.
\newblock Computational complexity of probabilistic turing machines.
\newblock {\em SIAM J. Comput.}, 6(4):675--695, 1977.

\bibitem{Jones1999}
Neil~D. Jones.
\newblock Logspace and ptime characterized by programming languages.
\newblock {\em Theoretical Computer Science}, 228:151--174, October 1999.

\bibitem{Jones2009a}
Neil~D. Jones and Lars Kristiansen.
\newblock A flow calculus of mwp-bounds for complexity analysis.
\newblock {\em ACM Trans. Comput. Logic}, 10(4):28:1--28:41, August 2009.

\bibitem{Kristiansen2005}
Lars Kristiansen and Neil~D. Jones.
\newblock The flow of data and the complexity of algorithms.
\newblock In {\em Proceedings of the First international conference on
  Computability in Europe: new Computational Paradigms}, CiE'05, pages
  263--274, Berlin, Heidelberg, 2005. Springer-Verlag.

\bibitem{Leivant1993}
Daniel Leivant.
\newblock Stratified functional programs and computational complexity.
\newblock In {\em Principles of Programming Languages, 20th International
  Symposium, Proceedings}, pages 325--333. ACM, 1993.

\bibitem{Leivant1995}
Daniel Leivant and Jean-Yves Marion.
\newblock Ramified recurrence and computational complexity {II}: Substitution
  and poly-space.
\newblock In Leszek Pacholski and Jerzy Tiuryn, editors, {\em Computer Science
  Logic, 9th International Workshop, Proceedings}, volume 933 of {\em LNCS},
  pages 486--500. 1995.

\bibitem{MMDiceMultipol}
R.~Metnani and J.-Y. Moyen.
\newblock Equivalence between the $mwp$ and {Q}uasi-{I}nterpretations analysis.
\newblock In J.-Y. Marion, editor, {\em DICE'11}, April 2011.

\bibitem{Meyer1967}
Albert~R. Meyer and Dennis~M. Ritchie.
\newblock The complexity of loop programs.
\newblock In {\em Proceedings of the 1967 22nd national conference}, ACM '67,
  pages 465--469, New York, NY, USA, 1967. ACM.

\end{thebibliography}
\end{document}